\documentclass[letterpaper, 10 pt, conference]{ieeeconf}  

\IEEEoverridecommandlockouts                              
\usepackage{adjustbox}
\usepackage{xcolor}
\usepackage{float}
\usepackage{amsmath}

\usepackage{amsthm,amssymb,graphicx,url}
\let\labelindent\relax 
\usepackage{enumitem}
\usepackage{booktabs}
\usepackage{caption}
\usepackage{subcaption}

\usepackage{appendix}

\usepackage{wrapfig}
\usepackage{hyperref}
\usepackage{soul}  
\usepackage{cleveref}
\usepackage{bm}
\usepackage{multicol}
\usepackage{mathtools}
\usepackage{cite}
\usepackage{graphicx}

\crefname{equation}{}{} 
\crefname{assumption}{Assumption}{}
\crefname{table}{Table}{} 
\crefname{figure}{Fig.}{}
\crefname{section}{Section}{}
\crefname{remark}{Remark}{}
\usepackage{algorithm,algorithmicx} 
\usepackage{algpseudocode} 
\newlength\myindent
\def\trieq{\triangleq}

\newtheorem{theorem}{Theorem}
\newtheorem{lemma}{Lemma}
\theoremstyle{definition}  
\theoremstyle{definition} \newtheorem{assumption}{Assumption}
\theoremstyle{remark}  \newtheorem{remark}{Remark}

\title{\LARGE \bf
Coordinated Path Following of UAVs using Event-Triggered Communication over Time-Varying Networks with Digraph Topologies
}

\author{Hyungsoo Kang, Isaac Kaminer, Venanzio Cichella, and Naira Hovakimyan
\thanks{This work is supported by AFOSR, NASA, ONR and NPS CRUSER.}
\thanks{Hyungsoo Kang and Naira Hovakimyan are with the Department of Mechanical Science and Engineering, University of Illinois at Urbana-Champaign, 
        Urbana, IL 61801, USA.  
        {\tt\small \{hk15, nhovakim\} @illinois.edu}}
\thanks{Isaac Kaminer is with the Department of
Mechanical and Aerospace Engineering, Naval Postgraduate School, Monterey, CA 93943, USA.
        {\tt\small kaminer@nps.edu}}
\thanks{Venanzio Cichella is with the Department of Mechanical
Engineering, University of Iowa, 
        Iowa City, IA 52242, USA.
        {\tt\small venanzio-cichella@uiowa.edu}} 
}

\begin{document}

\maketitle
\thispagestyle{empty}
\pagestyle{empty}

\begin{abstract}
In this article, a novel time-coordination algorithm based on event-triggered communications is proposed to achieve coordinated path-following of UAVs. To be specific, in the approach adopted a UAV transmits its progression information over a time-varying network to its neighbors only when a decentralized trigger condition is satisfied, thereby significantly reducing the volume of inter-vehicle communications required when compared with the existing algorithms based on continuous communications. Using such intermittent communications, it is shown that a decentralized coordination controller guarantees exponential convergence of the coordination error to a neighborhood of zero. Also, a lower bound on the interval between two consecutive event-triggered times is provided showing that the chattering issue does not arise with the proposed algorithm. Finally, simulation results validate the efficacy of the proposed algorithm.
\end{abstract}

\section{INTRODUCTION}
Recent remarkable progress in theory and technologies of multi-agent systems has given rise to a widespread use of multi-UAV systems. Relevant applications include collaborative payload transportation \cite{Lee2013,Lee2017}, formation flying for space exploration \cite{Morgan20141725,Bandyopadhyay2015} and cooperative SLAM \cite{Schmuck2017,Dubé2017} to name a few.

Among the diverse algorithms, coordinated path-following control has played a key role in solving challenging problems such as 1) search and rescue missions, where multiple UAVs progress in coordination covering a large area at a time; 2) sequential safe auto-landing, where multiple UAVs arrive at the glide slope safely separated by a predefined time interval; 3) atmospheric-science missions, where multiple UAVs fly cooperatively collecting data from a region of interest. 
The key requirement that the coordinated path following framework was developed to address was to guarantee a simultaneous arrival of each UAV at the end of its desired trajectory. The framework includes three steps: 1)  Generate  a set of collision-free desired trajectories to be assigned to each UAV that minimize a given cost and satisfy the simultaneous arrival requirement, boundary conditions, UAV dynamic constraints and guarantee obstacle avoidance \cite{Choe20161744}, \cite{Cichella2021}. 
2) Design  a  path-following control law \cite{CICHELLA201313} that steers a UAV along its desired trajectory.
3) Develop a decentralized time-coordination algorithm that enables each UAV to transmit progression information along its desired trajectory to its neighbors and also to adjust its progression speed based on the information provided by the neighbors. This step guarantees simultaneous time arrival by all the UAVs in the presence of disturbances.  

It has been shown in the early-stage research on coordinated path following that in Step 3 above the exponential stability of the time-coordination algorithms employed by all the UAVs  can be reduced to a consensus problem. Moreover, in, for example, \cite{Kaminer2006}, \cite{Ghabcheloo2007133} it was assumed that the topology of the underlying communication network  is represented by a connected  bidirectional graph with a fixed topology. However, this is a strong assumption less likely to be satisfied by a typical communication network with
time-varying network topology. To address this issue,  researchers in \cite{Xargay2013499}, \cite{Cichella2015945} were able to guarantee convergence of the time-coordination algorithms for the case where the network topology is represented by a time-varying bidirectional graph that is connected in an integral sense, i.e. the integral of the graph from $t$ to $t+T$ is connected $\forall t\geq0$ with a constant $T>0$. This condition, connectedness in an integral sense, was used in \cite{Mehdi2017}, \cite{Tabasso2020436} to show that collision avoidance can be achieved as well as the time-coordination. In the work reported in \cite{Tabasso2022704} the researchers applied the time-coordination algorithm \cite{Cichella2015945} to the problem of continuous monitoring of a path-constrained moving target. Our recent work \cite{Hyungsoo2023} showed convergence of time-coordination algorithms for a more general connectivity condition: the communication network is no longer required to be bidirectional. In fact, we have shown that the connectedness of the directed graph in an integral sense is sufficient to achieve  convergence. We note, however, that all the  algorithms discussed rely on continuous or piecewise continuous communications which may be undesirable or unavailable in real-world applications. Typical examples include underwater robotic missions where communication bandwidth is severely limited and military applications where stealth is of paramount importance. 

This issue can be overcome by resorting to event-triggered communication and control algorithms. With this approach, the communication and control input updates occur only when a predefined condition is satisfied, thereby significantly reducing the inter-vehicle communications in the process of achieving consensus on the variables of interest. It was pioneered by \cite{Dimarogonas2012,SEYBOTH2013245}. In most of the current work, the network topology is assumed to be a static bidirectional graph \cite{Dimarogonas2012,SEYBOTH2013245,Fan2015,YANG2019129} or digraph \cite{Hu2018,Qian2019}. Recently, some algorithms have been proposed where the network was assumed to be a time-varying bidirectional graph \cite{Wu2018,Cheng2019,Hu2019} or digraph \cite{Jia2018,HAN2015196,Hao2023}. Importantly, the existing consensus algorithms based on event-triggered communications have been developed in the absence of requirements on the change rate of the state, so they cannot be directly applied to the time-coordination problem: achieving state consensus $\gamma_i(t)=\gamma_j(t)$ in the presence of the requirement that $\dot{\gamma}_i(t)$ tracks a given desired rate $\dot{\gamma}_d(t)$. This gap motivated our research.

The contributions of this paper are summarized next. We propose a novel event-triggered communication algorithm that achieves the time-coordination objectives.
The topology of the underlying network is assumed to be a time-varying digraph connected in an integral sense. Each UAV transmits its progression information only when a decentralized trigger condition is satisfied, thereby significantly reducing the inter-vehicle communication requirements, particularly  when compared with the existing time-coordination algorithms based on continuous communications, see for example \cite{Kaminer2006,Ghabcheloo2007133,Xargay2013499,Cichella2015945}. We employ Lyapunov analysis to show that a proposed decentralized coordination controller guarantees exponential convergence of the time-coordination algorithms. Also, it is proven that the time interval between two consecutive events is bounded below. This implies that chattering is not an issue for the proposed algorithm. Since the algorithm is designed without making assumptions on the vehicle dynamics, it is applicable to any vehicle endowed with a path-following controller. In this paper, we illustrate the use of the algorithm  on quadrotors, a very widely used UAV. 

The rest of this article is organized as follows. Section~\ref{prelim} provides a brief review of graph theory. Section~\ref{III} formulates the time-coordination problem and presents assumptions on the inter-vehicle information flow. In Section~\ref{IV}, the time-coordination control law based on event-triggered communication is described. The performance of the time-coordination algorithm is analyzed in the main theorem. In Section~\ref{V}, simulation results validate the efficacy of the proposed algorithm. Finally, Section~\ref{VI} presents some conclusions.

\section{PRELIMINARIES} \label{prelim}
\subsection{Graph Theory}

\noindent


A digraph of size $n$ is defined by $\mathcal{D}=(\mathcal{V},\mathcal{E},\mathcal{A})$, where $\mathcal{V}=\{1,\dots,n\}$ is the set of nodes, $\mathcal{E}$ is the set of edges, and $\mathcal{A}$ is the Adjacency matrix. An edge is denoted by an ordered pair $(i,j)$, which means information can be transmitted from node $j$ to node $i$. The Adjacency matrix $\mathcal{A}$ is constructed as follows: if $(i,j)\in \mathcal{E}$, one has $\mathcal{A}_{ij}=1$. Otherwise, $\mathcal{A}_{ij}=0$. The digraph $\mathcal{D}$ is represented by the Laplacian $L\trieq \Delta-\mathcal{A}$, where $\Delta$ is a diagonal matrix with $\Delta_{ii}\trieq\sum_{j=1,j\neq i}^{n}\mathcal{A}_{ij}$. The neighborhood of node $i$ is the set $\mathcal{N}_i\trieq\{j\in\mathcal{V}: (i,j)\in \mathcal{E}\}$. A directed path from node $i_s$ to node $i_0$ is a sequence of edges $(i_0,i_1)$, $(i_1,i_2)$, $\dots$, $(i_{s-1},i_s)$. The digraph $\mathcal{D}$ contains a directed spanning tree if there exists a node such that it can reach every other node via a directed path.

Given a time-varying digraph $\mathcal{D}(t)$, we consider the digraph represented by the integrated Laplacian $\int_{t}^{t+T}L(\tau)d\tau$. An edge $(i,j)$ in it is said to be a $\delta$-edge if $\int_{t}^{t+T}-L_{ij}(\tau)d\tau\geq\delta$. 
A path in it is said to be a $\delta$-path if every edge on the path is a $\delta$-edge.


\section{TIME-COORDINATED PATH-FOLLOWING FRAMEWORK} \label{III}
\subsection{Path Following of a Singe UAV}
A trajectory generation algorithm such as \cite{Choe20161744}, \cite{Cichella2021} produces a set of collision-free desired trajectories for $n$ UAVs
\begin{align} \label{trajectory}
    p_{d,i}(t_d): [0,t_f]\rightarrow \mathbb{R}^3, \ \ i\in\{1,\dots,n\},
\end{align}
where $t_f$ is the simultaneous time of arrival. With the introduction of an adjustable nondecreasing function $\gamma_i(t)$ called coordination state or virtual time
\begin{align*}
    \gamma_i(t): [0,\infty)\rightarrow[0,t_f], \ \ \ i\in\{1,\dots,n\}, 
\end{align*}
we can denote the desired position of the $i$th quadrotor as $p_{d,i}(\gamma_i(t))$. This simple idea introduces an additional degree of freedom that allows one to adjust the progression speed of the UAV along its trajectory by controlling $\gamma_i(t)$ and thus to guarantee simultaneous arrival at $t=t_f$ in the presence of disturbances. The control law for $\gamma_i(t)$ will be presented in Section~\ref{IV}. 

To have the UAV track $p_{d,i}(\gamma_i(t))$, a path-following control law is needed. This can be done by driving the path following error 
\begin{align*}
    e_{PF,i}(t)\trieq p_i(t)-p_{d,i}(\gamma_i(t)), \ \ \ i\in\{1,\dots,n\}, 
\end{align*}
to zero. In the above expression, $p_i(t)$ represents the actual position of the UAV. In \cite{CICHELLA201313}, the authors formulated a path-following control law which ensures that the error converges exponentially to zero with ideal performance of the inner-loop autopilot and to a neighborhood of zero with non-ideal one. In other words, in the latter case, there exists $\rho>0$ such that 
\begin{align} \label{pf_error}
    \|e_{PF}(t)\|\leq\rho, \ \ \ \forall t\geq0,
\end{align}
where $e_{PF}(t)=[e_{PF,1}(t)^\top,\dots,e_{PF,n}(t)^\top]^\top$.

\subsection{Time Coordination of Multiple UAVs}
As discussed in the previous subsection, the progression of the UAV along its trajectory can be controlled by adjusting $\gamma_i(t)$. 
To be more specific, we impose the following  objectives. 

The UAVs are said to be synchronized at time $t$, if 
\begin{align} \label{obj1}
    \gamma_i(t)=\gamma_j(t), \ \ \ \forall i,j\in\{1,\dots,n\}.
\end{align}
Furthermore, for a desired mission progression pace $\dot{\gamma}_d(t)>0$, if 
\begin{align} \label{obj2}
    \dot{\gamma}_i(t)=\dot{\gamma}_d(t), \ \ \ \forall i\in\{1,\dots,n\},
\end{align}
the UAVs are considered as progressing in accordance to the desired progression mission pace. 

To satisfy the above requirements, UAVs need to exchange the coordination states $\gamma_i(t)$ with their neighbors over a time-varying directed network. 
This interaction among the UAVs can be modeled by using graph theory, see Section \ref{prelim}. 

The assumptions on the nature of the inter-vehicle communications used in this paper are provided next.

\begin{assumption} \label{assum1}
The information flow between any two UAVs is directional without time delays.
\end{assumption}

\begin{assumption} \label{assum2}
The $i$th UAV can receive coordination state $\gamma_j(t)$ only 
from other UAVs in its neighborhood set $\mathcal{N}_i(t)$, where $j\in\mathcal{N}_i(t)$.
\end{assumption}

The topology of the underlying communication network $\mathcal{D}(t)$ 
described by the Laplacian $L(t)$ varies in a way that satisfies the following assumption.

\begin{assumption} \label{assum3}
For all $t\geq0$, there exists $T>0$ such that the digraph represented by $\int_{t}^{t+T}L(\tau)d\tau$ contains a $\delta$-spanning tree. That is, a root node in it can reach every other node via a $\delta$-path. 
\end{assumption}

\begin{remark}
    The network Quality of Service (QoS) is determined by the parameters $T>0$ and $\delta \in (0,T]$. With a smaller value of $T$ and a value of $\delta$ closer to $T$, the~information flow has stronger connectivity. 
\end{remark}

\begin{remark}
    Notice that  Assumption~\ref{assum3} requires only connectedness in an integral sense, not pointwise in time. Thus, even if the network is disconnected for a certain period of time or at all times, this assumption can still be satisfied. 
\end{remark}

\noindent \textit{Problem (Time-Coordinated Path-Following Problem):} Given a set of desired trajectories \eqref{trajectory} and a path-following control law that ensures \eqref{pf_error}, design a decentralized time-coordination control law such that the coordination states $\gamma_i(t)$  converge for all $i$ exponentially near the equilibrium \eqref{obj1} and \eqref{obj2} under Assumptions~\ref{assum1}, \ref{assum2}, and \ref{assum3}.

\begin{remark}
    Designing a time-coordination control law for the coordination states $\gamma_i(t)$ that solves the Time-Coordinated Path-Following Problem guarantees the simultaneous arrival of all the UAVs at their respective destinations.
\end{remark}
\section{MAIN RESULT} \label{IV}
In this section, a decentralized time-coordination algorithm based on event-triggered communication (ETC) is presented. 

We propose the following decentralized control law 
\begin{align} \label{dyn1}
    \ddot{\gamma}_i(t)&=-b(\dot{\gamma}_i(t)-\dot{\gamma}_d(t)) \nonumber \\
    &\mathrel{\phantom{=}}-a\sum_{j\in\mathcal{N}_i(t)}(\gamma_i(t)-\hat{\gamma}_j(t))+\bar{\alpha}_i(e_{PF,i}(t)), \\
    \gamma_i(0)&=\gamma_{i0}, \ \ \dot{\gamma}_i(0)=\dot{\gamma}_{i0}, \nonumber
\end{align}
where $a$ and $b$ are positive coordination control gains and $\bar{\alpha}_i(e_{PF,i}(t))$ is defined as 
\begin{align} \label{alpha}
    \bar{\alpha}_i(e_{PF,i}(t))=\frac{\dot{p}_{d,i}(\gamma_i(t))^\top e_{PF,i}(t)}{\|\dot{p}_{d,i}(\gamma_i(t))\|+\eta}
\end{align}
with $\eta$ being a positive design parameter. 
\begin{remark}
    The term $\bar{\alpha}_i(e_{PF,i}(t))$ is designed in a way that  helps the UAV remain inside the region of attraction of the path-following controller \cite{CICHELLA201313}. When it precedes (falls behind of) the desired place, the numerator becomes positive (negative), which accelerates (decelerates) progression of the desired position along the trajectory leading to reduction in $\|e_{PF,i}(t)\|$. 
\end{remark}
In \eqref{dyn1}, an estimate $\hat{\gamma}_j(t)$ of $\gamma_j(t)$ is used because the~$i$th UAV intermittently receives $\gamma_j(t)$ from the $j$th UAV employing ETC. Therefore, it is a reasonable strategy for the~$i$th UAV to estimate $\hat{\gamma}_j(t)$ and to use this estimate to adjust its coordination state $\gamma_i(t)$. The estimator of $\gamma_j(t)$ is given below: 
\begin{equation} \label{estimator}
\hat{\gamma}_j(t):\left\{ \begin{aligned} 
  &\ddot{\hat{\gamma}}_j(t)=-b(\dot{\hat{\gamma}}_j(t)-\dot{\gamma}_d(t)), \ \ t\geq t^j_{k_j(t)} \\
  &\dot{\hat{\gamma}}_j(t^j_{k_j(t)})=\dot{\gamma}_j(t^j_{k_j(t)}), \ \ \hat{\gamma}_j(t^j_{k_j(t)})=\gamma_j(t^j_{k_j(t)})
\end{aligned} \right.,
\end{equation}
where $t^j_{k_j(t)}$ denotes the most recent 
time the $j$th UAV transmitted its coordination state ${\gamma}_j(t)$. 

Let the estimation error be $e_j(t)\trieq\hat{\gamma}_j(t)-\gamma_j(t)$. Then \eqref{dyn1} can be rewritten as follows
\begin{align} \label{dyn1'}
    \ddot{\gamma}_i(t)&=-b(\dot{\gamma}_i(t)-\dot{\gamma}_d(t))-a\sum_{j\in\mathcal{N}_i(t)}(\gamma_i(t)-\gamma_j(t)) \nonumber \\
    &\mathrel{\phantom{=}}+a\sum_{j\in\mathcal{N}_i(t)}e_j(t)+\bar{\alpha}_i(e_{PF,i}(t)), \\
    \gamma_i(0)&=\gamma_{i0}, \ \ \dot{\gamma}_i(0)=\dot{\gamma}_{i0}. \nonumber
\end{align}
As it can be seen in \eqref{dyn1'}, the error $e_j(t)$ has an explicit impact on the time-coordination dynamics. The key point of ETC algorithm is to ensure that $|e_j(t)|$ remains bounded. Therefore, when $|e_j(t)|$ reaches a given threshold, i.e., a transmission event is triggered, the $j$th UAV  transmits the time instant $t$ at which this event occurred as well as $\gamma_j(t)$, $\dot{\gamma}_j(t)$. To be more specific, we define an event-triggering function $\delta_j(t)$ as
\begin{align}
    \delta_j(t)=|e_j(t)|-h(t),
\end{align}
where $h(t)=c_1+c_2e^{-a t}$, $c_1,c_2,a\in\mathbb{R}_{\geq 0}$ is referred to as a threshold function, and $c_1\leq h(t)\leq c_1+c_2$. Whenever $\delta_j(t)> 0$, sampling and transmission take~place with $|e_j(t)|$ reset to $0$. With this ETC algorithm, boundedness of the error is guaranteed: $|e_j(t)|\leq h(t)$.
\begin{remark}
    The $i$th UAV uses the estimator \eqref{estimator} to propagate the time-coordination controller \eqref{dyn1}. On the other hand, the $j$th UAV executes the same code to check whether the event triggering condition $|e_j(t)|>h(t)$ is satisfied.
\end{remark}
\begin{remark}
    Due to the time-varying nature of the underlying network, there might be no available transmission channels when sampling occurs. Then the UAV continues sampling the data irrespective of the network topology and transmits the most recently sampled data $(t^j_{k_j(t)},\gamma_j(t^j_{k_j(t)}),\dot{\gamma}_j(t^j_{k_j(t)}))$ when it has available transmission channels.  
\end{remark}
For the ease of stability analysis, we introduce the coordination error state $\xi_{TC}(t)=[\xi_1(t)^\top \ \xi_2(t)^\top]^\top$ with
\begin{equation}
\begin{aligned} \label{error}
    \xi_1(t)&=Q\gamma(t) \ \in \mathbb{R}^{n-1}, \\
    \xi_2(t)&=\dot{\gamma}(t)-\dot{\gamma}_d(t)1_n \ \in \mathbb{R}^{n},
\end{aligned}
\end{equation}
where $\gamma(t)=[\gamma_1(t), \dots, \gamma_n(t)]^\top$, and $Q\in\mathbb{R}^{(n-1)\times n}$ is a matrix that satisfies $Q1_n=0_{n-1}$ and $QQ^\top=\mathbb{I}_{n-1}$. 
\begin{remark}
    A matrix $Q_k\in\mathbb{R}^{(k-1)\times k}$, $k\geq2$, satisfying $Q_k1_k=0_{k-1}$ and $Q_k\left(Q_k\right)^\top=\mathbb{I}_{k-1}$, can be constructed recursively:
    \begin{align*}
        Q_k=
        \begin{bmatrix}
        \sqrt{\frac{k-1}{k}} & -\frac{1}{\sqrt{k(k-1)}}1_{k-1}^\top \\
        0 & Q_{k-1} \\
        \end{bmatrix}
    \end{align*}
    with initial condition $Q_2=[1/\sqrt{2} \ -1/\sqrt{2}]$. For notational simplicity, we denote $Q_n$ by $Q$, where $n$ is the number of the UAVs.
\end{remark}
It is shown in \cite[Lemma~7]{phdenric2013} that $Q^\top Q=\mathbb{I}_n-\frac{1_n1_n^\top}{n}$ and the nullspace of Q is spanned by $1_n$. Now suppose that $\xi_1(t)=Q\gamma(t)=0_{n-1}$. Then since the nullspace of Q is spanned by $1_n$, we obtain that $\gamma_i(t)=\gamma_j(t)$, $\forall i,j\in\{1,\dots,n\}$. Furthermore, $\xi_2(t)=0_n$ implies that 
$\dot{\gamma}_i(t)=\dot{\gamma}_d(t)$, $\forall i\in\{1,\dots,n\}$. Therefore, $\xi_{TC}(t)=0_{2n-1}$ is equivalent to \eqref{obj1} and \eqref{obj2}. 

Using definitions of $\gamma(t)$, $L(t)$ and $\mathcal{A}(t)$, the expression \eqref{dyn1'} can be rewritten in a more compact form
\begin{align} 
    \ddot{\gamma}(t)&=-b\xi_2(t)-aL(t)\gamma(t)+a\mathcal{A}(t)e(t)+\bar{\alpha}(e_{PF}(t)), \nonumber\\
    \gamma(0)&=\gamma_0, \ \ \dot{\gamma}(0)=\dot{\gamma}_0, \nonumber
\end{align} 
where $\bar{\alpha}(e_{PF}$ $(t))=[\bar{\alpha}_1(e_{PF,1}(t)),\dots,\bar{\alpha}_n(e_{PF,n}(t))]^\top$ and $e(t)=[e_1(t),\dots,e_n(t)]^\top$.
\begin{lemma} \label{lem1}
Consider the following dynamics
\begin{align} \label{lem1:equ}
    \dot{x}=-\frac{a}{b}L(t)x, \ \ \ x(0)=x_0\in\mathbb{R}^n.
\end{align}
Under Assumption~\ref{assum3} on $L(t)$, the components $x_i$'s of $x$ reach consensus exponentially
\begin{align}
    diam\left(x\left(t\right)\right)&\trieq\max\limits_{i}\{x_i(t)\}-\min\limits_{i}\{x_i(t)\} \\
    &\leq diam\left(x(0)\right)ke^{-\lambda t}, \label{lem2:equ2}
\end{align}
where $k\trieq\frac{1}{1-(\delta')^n}$ and $\lambda\trieq-\frac{1}{nT}\ln(1-(\delta')^n)$ with $\delta'\trieq \min\left\{1,\frac{a}{b}\delta\right\}e^{-(n-1)\frac{a}{b}T}$.
\end{lemma}
\begin{proof}
    The proof follows from the first half of the proof of \cite[Theorem~1]{HAN2015196}.
\end{proof}

\begin{lemma} \label{lem2}
The quantities $\|Qx\|$ and $diam(x)$ 
satisfy the following inequalities
\begin{align} \label{lem2:equ}
    \frac{1}{\sqrt{n}}\|Qx\|\leq diam(x)\leq \sqrt{2}\|Qx\|.
\end{align}
\end{lemma}
\begin{proof}
The first inequality in \eqref{lem2:equ} follows from 
\begin{align*}
    \resizebox{\hsize}{!}{$\|Qx\|^2=x^\top Q^\top Qx=x^\top \left(\mathbb{I}_n-\frac{1_n1_n^\top}{n}\right) x=n\left\{\frac{1}{n}\sum_{i=1}^{n}x_i^2-\left(\frac{1_n^\top x}{n}\right)^2\right\}$} \\
    \resizebox{\hsize}{!}{$=\sum_{i=1}^{n}\left(x_i-\frac{1_n^\top x}{n}\right)^2\leq n\left(\max\limits_{i}\{x_i\}-\min\limits_{i}\{x_i\}\right)^2=n\left\{diam(x)\right\}^2$}
\end{align*}
The second inequality in \eqref{lem2:equ} follows from 
\begin{align*}
    \resizebox{\hsize}{!}{$\left\{diam(x)\right\}^2=\left(\max\limits_{i}\{x_i\}-\min\limits_{i}\{x_i\}\right)^2=\left(\max\limits_{i}\{x_i\}-\frac{1_n^\top x}{n}+\frac{1_n^\top x}{n}-\min\limits_{i}\{x_i\}\right)^2$} \\
    \resizebox{\hsize}{!}{$\leq2\left(\max\limits_{i}\{x_i\}-\frac{1_n^\top x}{n}\right)^2+2\left(\frac{1_n^\top x}{n}-\min\limits_{i}\{x_i\}\right)^2\leq2\sum_{i=1}^{n}\left(x_i-\frac{1_n^\top x}{n}\right)^2=2\|Qx\|^2$}
\end{align*}
\end{proof} 
\begin{remark}
    When the elements $x_i$ of a vector $x$ reach consensus, i.e., $x_1=\cdots=x_n$, both $\|Qx\|$ and $diam(x)$ are zero. Otherwise, they have positive values, which become larger as $x_i$'s diverge away from consensus. In other words, $\|Qx\|$ and $diam(x)$ quantify the discoordination among $x_i$'s. Moreover, Lemma \ref{lem2} implies that $\|Qx\|$ and $diam(x)$ are equivalent measures of the discoordination among $x_i$'s. 
\end{remark}
The main results of the paper are presented in the following theorem.
\begin{theorem} \label{thm}
Consider a set of desired trajectories \eqref{trajectory} and a path-following controller that ensures \eqref{pf_error}. Let the evolution of $\gamma_i(t)$ be governed by \eqref{dyn1} over the network $\mathcal{D}(t)$ satisfying Assumption~\ref{assum3}. Then, there exist time coordination control gains $a$, $b$, and $\eta$ such that
\begin{equation}
\begin{aligned}
    \|\xi_{TC}(t)\|&\leq\kappa_1\|\xi_{TC}(0)\|e^{-\lambda_{TC}t} \\&\mathrel{\phantom{=}}+\kappa_2\sup_{t\geq0}\left(an\sqrt{n}h(t)+\|e_{PF}(t)\|+|\ddot{\gamma}_d(t)|\right) \nonumber
\end{aligned}
\end{equation}
with rate of convergence
\begin{align} \label{rate}
    \lambda_{TC}\leq\frac{\lambda}{6nk^2},
\end{align}
where $\lambda$ and $k$ were defined in Lemma~\ref{lem1}. 

Moreover, the event-triggered time interval $t^i_{k+1}-t^i_k$ is  bounded below. Here, $t^i_k$ denotes the $k$th event-triggered time of the $i$th UAV.
\end{theorem}
\begin{proof}
Motivated by \cite{Cichella2015945}, we introduce a new state
\begin{align*}
    \chi(t)=b\xi_1(t)+Q\xi_2(t).
\end{align*}
Then, the coordination error state $\xi_{TC}(t)=[\xi_1(t)^\top \ \xi_2(t)^\top]^\top$ can be redefined as $\bar{\xi}_{TC}(t)=[\chi(t)^\top \ \xi_2(t)^\top]^\top$ with dynamics
\begin{equation}
\begin{aligned} \label{dyn3}
    \dot{\chi}&=-\frac{a}{b}\bar{L}(t)\chi+\frac{a}{b}QL(t)\xi_2+aQ\mathcal{A}(t)e+Q\bar{\alpha}(e_{PF}) \\
    \dot{\xi}_2&=-\frac{a}{b}L(t)Q^\top\chi-\left(b\mathbb{I}_n-\frac{a}{b}L(t)\right)\xi_2+a\mathcal{A}(t)e \\
    &\mathrel{\phantom{=}}+ \ \bar{\alpha}(e_{PF})-\ddot{\gamma}_d1_n,
\end{aligned}
\end{equation}
where $\bar{L}(t)\trieq QL(t)Q^\top \in \mathbb{R}^{(n-1)\times(n-1)}$. \\
In order to construct a Lyapunov function candidate for \eqref{dyn3}, we first show that the following auxiliary system
\begin{align} \label{aux}
    \dot{\phi}(t)=-\frac{a}{b}\bar{L}(t)\phi(t), \ \ \ \phi(0)=\phi_0\in\mathbb{R}^{n-1}
\end{align}
is globally uniformly exponentially stable (GUES). As $Q\in\mathbb{R}^{(n-1)\times n}$ is a full rank matrix, there exists $x_0\in\mathbb{R}^n$ such that $\phi_0=Qx_0$. Let $x(t)$ be the solution of \eqref{lem1:equ}. Then $Qx(t)$ is a unique solution of \eqref{aux}:
\begin{align*}
    \dot{\phi}+\frac{a}{b}\bar{L}(t)\phi&=Q\dot{x}+\frac{a}{b}\bar{L}(t)Qx=Q\left(\dot{x}+\frac{a}{b}L(t)Q^\top Qx\right) \\
    &=Q\left(\dot{x}+\frac{a}{b}L(t)x\right)\equiv0,
\end{align*}
where the third equality follows from the fact that $L(t)Q^\top Q=L(t)$. Therefore, we have
\begin{align*}
    \|\phi\|&=\|Qx\|\leq\sqrt{n} \ diam(x)\leq \sqrt{n} \ diam(x_0) ke^{-\lambda t} \\
    &\leq \sqrt{n}\bigl(\sqrt{2}\|Qx_0\|\bigr)ke^{-\lambda t}= k_\phi\|\phi_0\|e^{-\lambda t},
\end{align*}
where $k_\phi\trieq\sqrt{2n}k$, the second inequality follows from \eqref{lem2:equ2}, and the other inequalities follow from \eqref{lem2:equ}. The GUES of \eqref{aux} can now be used to define a Lyapunov function candidate to analyze stability of \eqref{dyn3} as presented in the subsequent discussion. \\
Since the system \eqref{aux} is GUES, Theorem $4.12$ in \cite{kha2002} implies that there exists a continuously differentiable, symmetric, positive definite matrix $\Psi(t)$ such that
\begin{align}
    c_1\mathbb{I}_{n-1}\trieq \frac{bc_3}{2an}\mathbb{I}_{n-1}\leq \Psi(t) \leq \frac{k^2_\phi c_4}{2\lambda}\mathbb{I}_{n-1}\trieq c_2\mathbb{I}_{n-1}, \label{p1} \\
    \dot{\Psi}(t)-\frac{a}{b}\bar{L}^\top(t) \Psi(t)-\frac{a}{b}\Psi(t)\bar{L}(t)\leq -c_3\mathbb{I}_{n-1}, \label{p2}
\end{align}
where $c_3$ and $c_4$ are any constants satisfying $0<c_3\leq c_4$. \\
Using $\Psi(t)$, we define a Lyapunov function candidate for \eqref{dyn3} as follows:
\begin{align} \label{lya}
    V_{TC}(t)=\chi^\top \Psi(t)\chi+\frac{\beta}{2}\|\xi_2\|^2=\bar{\xi}^\top_{TC} W(t)\bar{\xi}_{TC},
\end{align}
where $\beta>0$ and $W(t)\trieq\begin{bmatrix}
\Psi(t) & 0 \\ 
0 & \frac{\beta}{2}\mathbb{I}_n
\end{bmatrix}$. Notice that $V_{TC}(t)$ satisfies $z^\top M_1z\leq V_{TC}(t)\leq z^\top M_2z$, where $z\trieq[\|\chi\| \ \|\xi_2\|]^\top$ with 
\begin{align*} 
M_1\trieq\begin{bmatrix}
c_1 & 0 \\ 
0 & \beta/2
\end{bmatrix} \text{ and }
M_2\trieq\begin{bmatrix}
c_2 & 0 \\ 
0 & \beta/2
\end{bmatrix}.
\end{align*}
The time derivative of \eqref{lya} along the trajectory of \eqref{dyn3} satisfies \\
\begin{adjustbox}{max width=\linewidth}
\parbox{\linewidth}{\begin{align*}
    \dot{V}_{TC}&=\chi^\top\left(\dot{\Psi}(t)-\frac{a}{b}\bar{L}^\top(t) \Psi(t)-\frac{a}{b}\Psi(t)\bar{L}(t)\right)\chi \\
    &\mathrel{\phantom{=}}-\beta\xi^\top_2\left(b\mathbb{I}_n-\frac{a}{b}L(t)\right)\xi_2 \\
    &\mathrel{\phantom{=}}+\chi^\top\left(2\frac{a}{b}\Psi(t)QL(t)-\beta\frac{a}{b}QL^\top(t)\right)\xi_2 \\
    &\mathrel{\phantom{=}}+\left(2\chi^\top \Psi(t)Q+\beta\xi^\top_2\right)\left(a\mathcal{A}(t)e+\bar{\alpha}(e_{PF})\right)-\beta\xi^\top_2\ddot{\gamma}_d1_n,
\end{align*}}
\end{adjustbox}
which leads to \\
\begin{adjustbox}{max width=\linewidth}
\parbox{\linewidth}{\begin{align*}
    \dot{V}_{TC}\leq& -c_3\|\chi\|^2-\beta\left(b-\frac{a}{b}n\right)\|\xi_2\|^2 \\
    &+\left(2\frac{a}{b}n\|\Psi(t)\|+\beta\frac{a}{b}n\right)\|\chi\|\|\xi_2\| \\
    &+\left(2\|\Psi(t)\|\|\chi\|+\beta\|\xi_2\|\right)\left(an\|e\|+\|\bar{\alpha}(e_{PF})\|+|\ddot{\gamma}_d|\right),
\end{align*}}
\end{adjustbox}
where we used \eqref{p2}, $\|Q\|=1$, $\|L(t)\|\leq n$, and $\|\mathcal{A}(t)\|\leq n$. \\
Applying $\|\Psi(t)\|\leq c_2=\frac{k^2_\phi c_4}{2\lambda}$ in \eqref{p1} yields \\
\begin{adjustbox}{max width=\linewidth}
\parbox{\linewidth}{\begin{align*}
    \dot{V}_{TC}\leq& -c_3\|\chi\|^2-\beta\left(b-\frac{a}{b}n\right)\|\xi_2\|^2 \\
    &+\left(\frac{a}{b}\frac{nk_\phi^2}{\lambda}c_4+\beta\frac{a}{b}n\right)\|\chi\|\|\xi_2\| \\
    &+\left(\frac{k^2_\phi c_4}{\lambda}+\beta\right)\|\bar{\xi}_{TC}\|\left(an\|e\|+\frac{v_{max}}{v_{min}+\eta}\|e_{PF}\|+|\ddot{\gamma}_d|\right),
\end{align*}}
\end{adjustbox}
where $v_{max}=\max_{i}\{v_{i,max}\}$ and $v_{min}=\max_{i}\{v_{i,min}\}$ with $v_{i,max}$ and $v_{i,min}$ being the maximum and minimum achievable speed of the $i$th UAV.
Letting $c_3=c_4$ and $\eta>v_{max}-v_{min}$, one obtains \\
\begin{adjustbox}{max width=\linewidth}
\parbox{\linewidth}{\begin{align*}
    \dot{V}_{TC}\leq -z^\top Uz+\left(\frac{k^2_\phi c_4}{\lambda}+\beta\right)\|\bar{\xi}_{TC}\|\left(an\|e\|+\|e_{PF}\|+|\ddot{\gamma}_d|\right),
\end{align*}}
\end{adjustbox}
where $z=[\|\chi\| \ \|\xi_2\|]^\top$ and
\begin{align*}
U\trieq\begin{bmatrix}
c_3 & -\frac{1}{2}\left(\frac{a}{b}\frac{nk_\phi^2}{\lambda}c_3+\beta\frac{a}{b}n\right) \\ 
-\frac{1}{2}\left(\frac{a}{b}\frac{nk_\phi^2}{\lambda}c_3+\beta\frac{a}{b}n\right) & \beta\left(b-\frac{a}{b}n\right)
\end{bmatrix}.
\end{align*}
Next, we introduce $\lambda_{TC}\leq\frac{2\lambda}{3k^2_\phi}=\frac{\lambda}{6nk^2}$ which defines the time-coordination error convergence rate. Consider
\begin{align} 
\label{UUU}
    &U-3\lambda_{TC}M_2 \nonumber \\
    &=
    \begin{bmatrix}
    c_3-\lambda_{TC} \frac{3k^2_\phi}{2\lambda}c_3 & -\frac{1}{2}\left(\frac{a}{b}\frac{nk_\phi^2}{\lambda}c_3+\beta\frac{a}{b}n\right) \\ 
    -\frac{1}{2}\left(\frac{a}{b}\frac{nk_\phi^2}{\lambda}c_3+\beta\frac{a}{b}n\right) & \beta\left(b-\frac{a}{b}n-\frac{3}{2}\lambda_{TC}\right)
    \end{bmatrix}. 
\end{align}
Note that for a fixed value of $\frac{a}{b}$, all the terms in \eqref{UUU} are fixed except for $\beta b$ in the $(2,2)$ element. This is because the values of $k_\phi$ and $\lambda$ are determined by the ratio of $\frac{a}{b}$. Thus, choosing a sufficiently large $b$ with a fixed $\frac{a}{b}$ makes sure that \eqref{UUU} is positive semi-definite. \\
Therefore, we obtain that $-z^\top Uz\leq-3\lambda_{TC}z^\top M_2z\leq-3\lambda_{TC}V_{TC}$ and thus the derivative of $V_{TC}$ is  bounded above by
\begin{align*}
    \dot{V}_{TC}&\leq -3\lambda_{TC}V_{TC} \\
    &\mathrel{\phantom{\leq}}+\left(\frac{k^2_\phi c_3}{\lambda}+\beta\right)\|\bar{\xi}_{TC}\|\left(an\|e\|+\|e_{PF}\|+|\ddot{\gamma}_d|\right) \\
    &\leq-2\lambda_{TC}V_{TC}-\lambda_{TC}\min\{c_1,\beta/2\}\|\bar{\xi}_{TC}\|^2 \\
    &\mathrel{\phantom{\leq}}+\left(\frac{k^2_\phi c_3}{\lambda}+\beta\right)\|\bar{\xi}_{TC}\|\left(an\|e\|+\|e_{PF}\|+|\ddot{\gamma}_d|\right).
\end{align*}
By applying Lemma $4.6$ in \cite{kha2002} and introducing the state transformation
$\bar{\xi}_{TC}=S\xi_{TC}\trieq
\begin{bmatrix}
b\mathbb{I}_{n-1} & Q \\ 
0 & \mathbb{I}_n
\end{bmatrix}\xi_{TC}$, we conclude that

\begin{equation} \label{iss}
\begin{aligned}
    \|\xi_{TC}(t)\|&\leq\kappa_1\|\xi_{TC}(0)\|e^{-\lambda_{TC}t} \\
    &\mathrel{\phantom{\leq}}+\kappa_2\sup_{t\geq0}\left(an\sqrt{n}h(t)+\|e_{PF}(t)\|+|\ddot{\gamma}_d(t)|\right), 
\end{aligned}
\end{equation}
where $\|e(t)\|\leq \sqrt{n}h(t)$ was used and
\begin{align}
    \kappa_1&\trieq\|S^{-1}\|\sqrt{\frac{\max\{c_2,\beta/2\}}{\min\{c1,\beta/2\}}}\|S\|, \label{kappa1} \\
    \kappa_2&\trieq\|S^{-1}\|\sqrt{\frac{\max\{c_2,\beta/2\}}{\min\{c1,\beta/2\}}}\frac{\frac{k^2_\phi c_3}{\lambda}+\beta}{\lambda_{TC}\min\{c_1,\beta/2\}}. \label{kappa2} 
\end{align}
Lastly, we show that $t^i_{k+1}-t^i_k$ is  bounded below. From \eqref{estimator} and \eqref{dyn1'}, it follows that the estimation error dynamics can be written in the form of
\begin{equation}
\begin{aligned}
    \dot{\epsilon}_i(t)&=A\epsilon_i(t)+Bu_i(t), \ \ \ t\in[t^i_k,t^i_{k+1}), \label{dyn:error}\\
    \epsilon_i(t^i_k)&=[0 \ 0]^\top,
\end{aligned}
\end{equation}
where $\epsilon_i(t)=[e_i(t) \ \dot{e}_i(t)]^\top$, $A=\begin{bmatrix}
0 & 1 \\
0 & -b \\
\end{bmatrix}$, $B=\begin{bmatrix}
0 \\
1 \\
\end{bmatrix}$, and $u_i=a\sum_{j\in\mathcal{N}_i}(\gamma_i-\gamma_j)-a\sum_{j\in\mathcal{N}_i}e_j-\bar{\alpha}_i(e_{PF,i})$.

\noindent One can show that \\
\begin{adjustbox}{max width=\linewidth}
\parbox{\linewidth}{\begin{align*}
    |u_i|&\leq a\|L\gamma\|+a\|\mathcal{A}e\|+\|\bar{\alpha}(e_{PF})\| \\
    &\leq a\|L\|\|Q^\top\|\|Q\gamma\|+a\|\mathcal{A}\|\|e\|+\|e_{PF}\| \\
    &\leq an\left(\kappa_1\|\xi_{TC}(0)\|+\kappa_2\sup_{t\geq0}\left(an\sqrt{n}h(t)+\|e_{PF}\|+|\ddot{\gamma}_d|\right)\right) \\
    &\mathrel{\phantom{\leq}}+a\|\mathcal{A}\|\|e\|+\|e_{PF}\| \\
    &\leq an\kappa_1\|\xi_{TC}(0)\|+an\kappa_2\left(an\sqrt{n}\left(c_1+c_2\right)+\rho+\ddot{\gamma}_{d,max}\right) \\
    &\mathrel{\phantom{\leq}}+an\sqrt{n}(c_1+c_2)+\rho\trieq\bar{u}.
\end{align*}}
\end{adjustbox}
Therefore, we can show that the error $e_i(t)$ is bounded above using \eqref{dyn:error} as follows:
\begin{equation}
\begin{aligned}
    |e_i|\leq\|\epsilon_i\|&\leq\int^t_{t^i_k}(\|A\|\|\epsilon_i\|+\|B\||u_i|)d\tau \\
    &\leq\int^t_{t^i_k}\|A\|\|\epsilon_i\|d\tau+(t-t^i_k)\|B\|\bar{u}.
    \label{GB}
\end{aligned}
\end{equation}
Applying Gronwall-Bellman inequality, see Lemma A.1 in \cite{kha2002}, to  inequality \eqref{GB} leads to \\
\begin{adjustbox}{max width=\linewidth}
\parbox{\linewidth}{\begin{align*}
    |e_i|&\leq\|\epsilon_i\|\leq(t-t^i_k)\|B\|\bar{u}+\|A\|\int^t_{t^i_k}(s-t^i_k)\|B\|\bar{u}e^{\|A\|(t-s)}ds \\
    &\leq\|B\|\bar{u}\left( e^{\|A\|(t-t^i_k)}-1\right)/\|A\|.
\end{align*}}
\end{adjustbox}
Considering that the transmission event is triggered when $|e_i(t)|>h(t)\geq c_1$, the inter-event time interval is bounded below:
\begin{align*}
    t^i_{k+1}-t^i_k\geq\frac{1}{\|A\|}\ln{(1+c_1\|A\|/(\|B\|\bar{u}))}>0.
\end{align*}
This completes the proof of Theorem~\ref{thm}.
\end{proof}

\begin{remark}
    We note that the rate of convergence $\lambda_{TC}$ in \eqref{rate} is determined by $T>0$ and $\delta\in(0,T]$, which represent the QoS of the underlying network and also by the values of the coordination control gains $a$ and $b$.
\end{remark}

\section{SIMULATION RESULTS} \label{V}
This section presents simulation results of a coordinated path-following mission, which illustrate the efficacy of the proposed algorithm. The trajectory generation algorithm \cite{Cichella2021} designs a set of Bezier curves with the following specification. The starting points are on $y=0\,m$. The trajectories simultaneously reach $y=150\,m$ exchanging $(x,z)$ coordinates. The inter-vehicle safety distance is $10\,m$. The mission duration is $t_f=21.10\,s$. The solid curves in Figure~\ref{fig:traj} depict the desired trajectories.
\begin{figure} [h!]
    \centering
    \includegraphics[width = 1.00\linewidth]{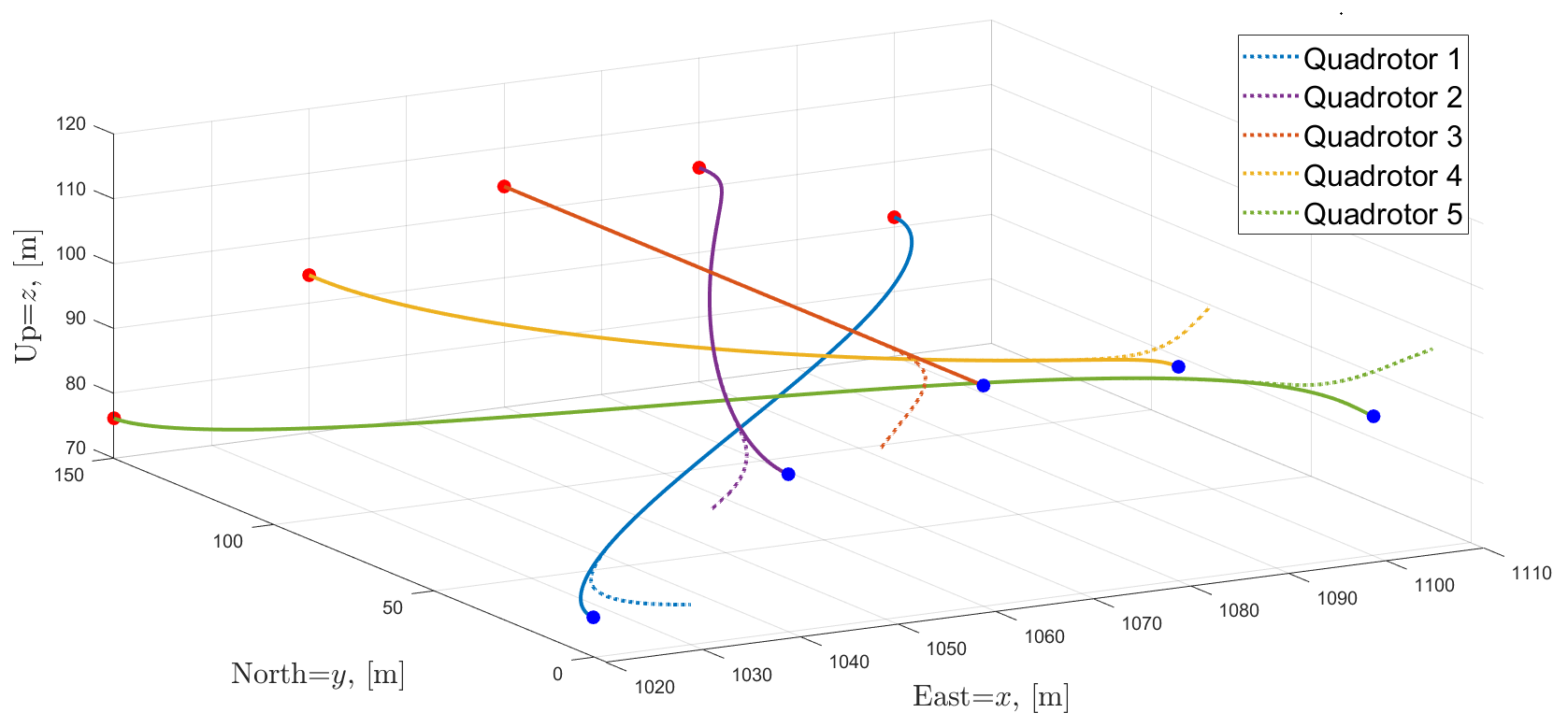}
    \caption{Time-coordinated path-following of five quadrotors. The starting points of the desired trajectories, blue dots, are on $y=0\,m$. The final points, red dots, are on $y=150\,m$.}
    \label{fig:traj}
\end{figure}

Considering a more general situation, the initial positions of the quadrotors are not the same as the initial points of the desired trajectories, i.e., they have initial path-following errors. The path-following controller \cite{CICHELLA201313} steers each quadrotor to the desired trajectory. The dotted curves in Figure~\ref{fig:traj} are the paths travelled by the quadrotors.

Suppose that the quadrotors activate transmission edges in the order of $\mathcal{D}_1$ $\rightarrow$ $\mathcal{D}_2$ $\rightarrow$  $\mathcal{D}_3$ $\rightarrow$  $\mathcal{D}_1$ $\rightarrow$ $\cdots$ in Figure~\ref{fig:d}, with the duration of each $0.03\,s$ due to tight communication bandwidth. Even though $\mathcal{D}(t)$ is not connected at all times, the integrated graph represented by $\int_{t}^{t+0.09}L(\tau)d\tau$ $(\forall t\geq0)$ contains $0.03$-spanning tree, that is, connected in an integral sense satisfying the Assumption~\ref{assum3}. The coordination control gains and the parameter $\eta$ in \eqref{alpha} are set to $a=3.75$, $b=4.82$, and $\eta=12$. The initial conditions for the coordination state are $\gamma(0)=0_n$ and $\dot{\gamma}(0)=1_n$.
\begin{figure} [h!]
     \centering
     \begin{subfigure}[h]{0.1\textwidth}
         \centering
         \includegraphics[width=\textwidth]{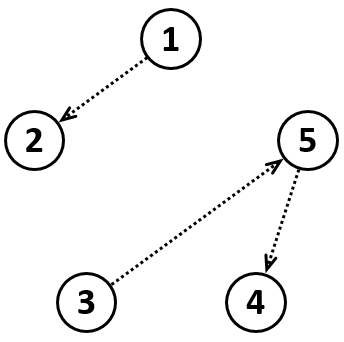}
         \caption{$\mathcal{D}_1$}
         \label{fig:d1}
     \end{subfigure}
     \hspace{2.5em}
     \begin{subfigure}[h]{0.1\textwidth}
         \centering
         \includegraphics[width=\textwidth]{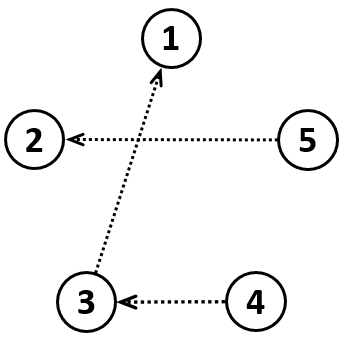}
         \caption{$\mathcal{D}_2$}
         \label{fig:d2}
     \end{subfigure}
     \hspace{2.5em}
     \begin{subfigure}[h]{0.1\textwidth}
         \centering
         \includegraphics[width=\textwidth]{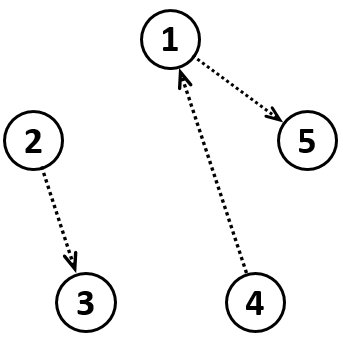}
         \caption{$\mathcal{D}_3$}
         \label{fig:d3}
     \end{subfigure}
        \caption{Switching topology of available transmission edges.}
        \label{fig:d}
\end{figure} 

Figures~\ref{fig:delta_gamma}-\ref{fig:event} illustrate how the coordination controller~\eqref{dyn1} works in combination with \eqref{estimator}. Initially, quadrotors $3$ and $5$ lie ahead of the plane $y=0\,m$, quadrotor $4$ lies behind it, and the remaining ones are on it. By the definition of~\eqref{alpha}, the term $\bar{\alpha}_i(e_{PF,i}(0))$ is positive for $i=3,5$; negative for $i=4$; and close to zero for the remaining $i$'s. As $\bar{\alpha}_i(e_{PF,i}(t))$ is added to the right hand side of \eqref{dyn1}, it makes sense that $\gamma_3(t)$, $\gamma_5(t)$ accelerate and $\gamma_4(t)$ decelerates for the first few seconds in Figure~\ref{fig:gamma_dot}. This evolution of $\gamma_i(t)$ helps each quadrotor to quickly approach the desired position remaining inside the region of attraction, however at the cost of increased inter-vehicle discoordination. Here, the discoordination is alleviated with the proposed ETC algorithm. Comparing \eqref{dyn1} with \eqref{estimator}, we can notice that the estimation error $e_i(t)=\hat{\gamma}_i(t)-\gamma_i(t)$ stems from $\bar{\alpha}_i(e_{PF,i}(t))$. With larger values of $\bar{\alpha}_i(e_{PF,i}(0))$ for $i=3,4,5$, the corresponding error $e_i(t)$, $i=3,4,5$ evolves fast and whenever it reaches the the threshold function $h(t)$, sampling and transmission takes place. In our simulation, we used $h(t)=0.03$. In Figure~\ref{fig:event}, it is confirmed that quadrotors $3$, $4$, and $5$ more frequently take a sample and transmit it than the others do. The transmitted data is used by the estimator \eqref{estimator} to compute the estimate, which is then used in the second term of \eqref{dyn1} to achieve the inter-vehicle coordination. The fleet simultaneously arrives on $y=150\,m$ at $t=21.06\,s$, slightly earlier than the original schedule $t_f=21.10\,s$. It is because the quadrotors $3$ and $5$, initially lying ahead of $y=0\,m$, made the mission unfold fast. In this case, it is possible to have the fleet arrive on $y=150\,m$ as planned by adjusting $\dot{\gamma}_d(t)$ from $1$ to a smaller value in the middle of the mission. 

Notice that once the coordination errors become sufficiently small, the event is triggered less frequently, Figure~\ref{fig:event}. Therefore, we can say that the inter-vehicle communication occurs intelligently only when discoordination between the quadrotors increases. Lastly, in the case when the quadrotors are discoordinated by disturbances such as wind gusts, they can restore the coordination in the same manner.

\begin{figure} [h!]
    \centering
    \includegraphics[width = 1.00\linewidth]{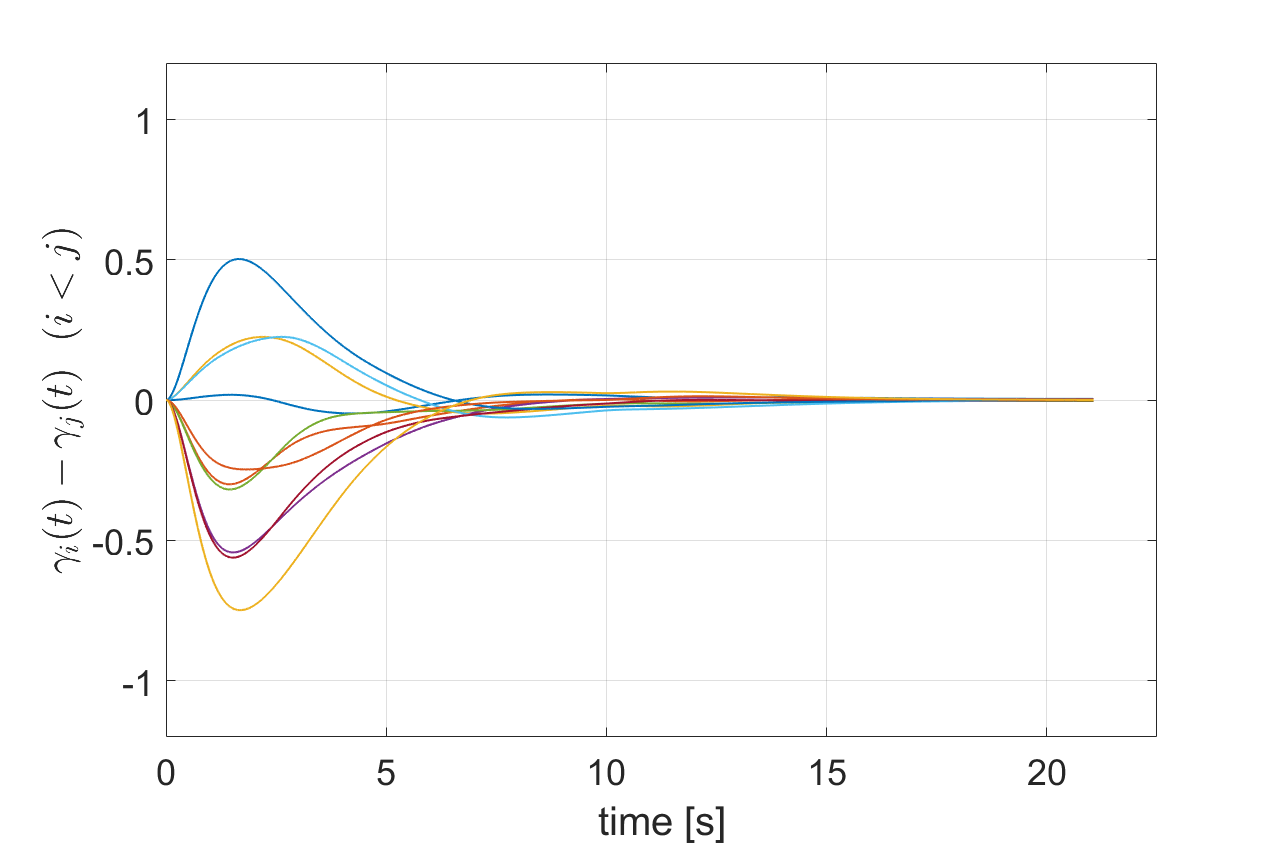}
    \caption{Convergence of the coordination error $\gamma_i(t)-\gamma_j(t)$ $(i<j)$ to a neighborhood of zero.}
    \label{fig:delta_gamma}
\end{figure}

\begin{figure} [h!]
    \centering
    \includegraphics[width = 1.00\linewidth]{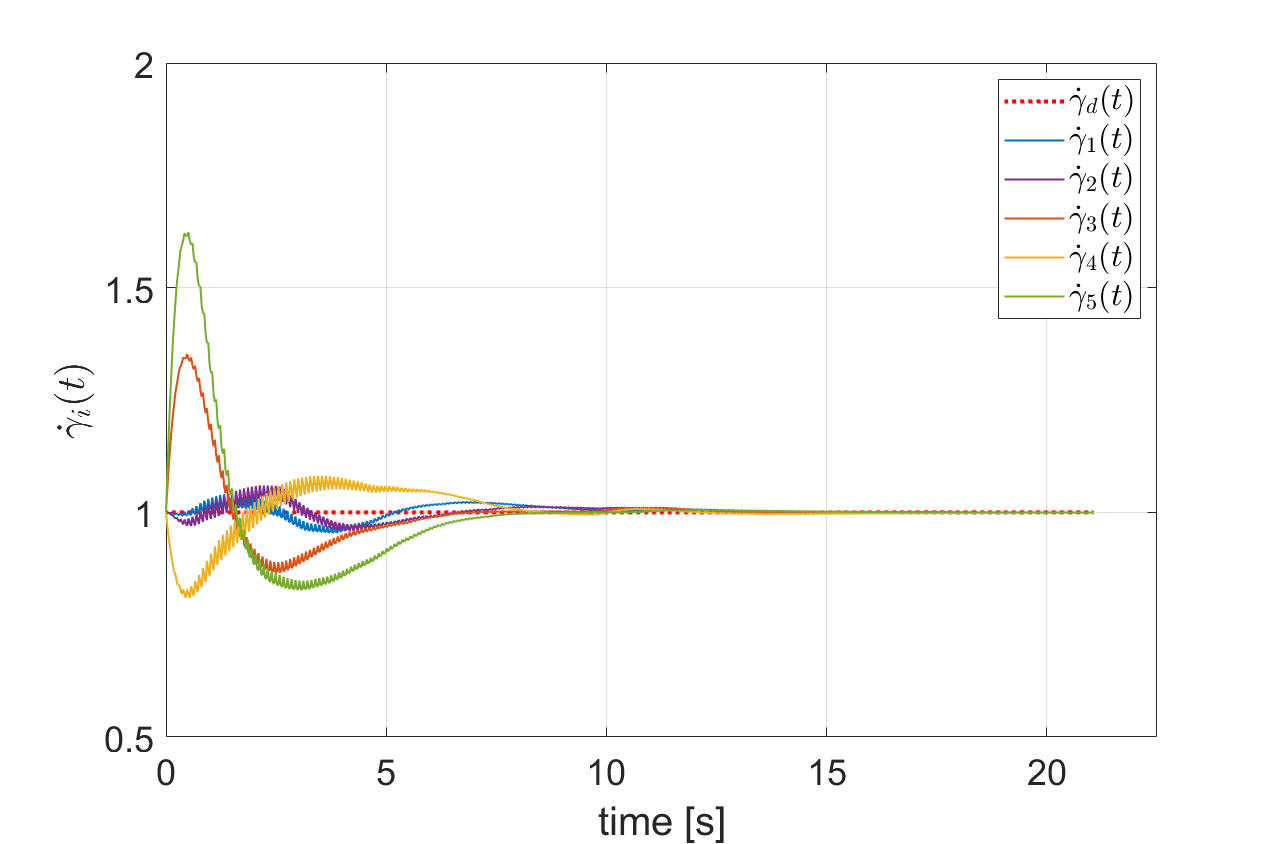}
    \caption{Convergence of $\dot{\gamma}_i(t)$ to a neighborhood of $\dot{\gamma}_d(t)=1$.}
    \label{fig:gamma_dot}
\end{figure}

\begin{figure} [h!]
    \centering
    \includegraphics[width = 1.00\linewidth]{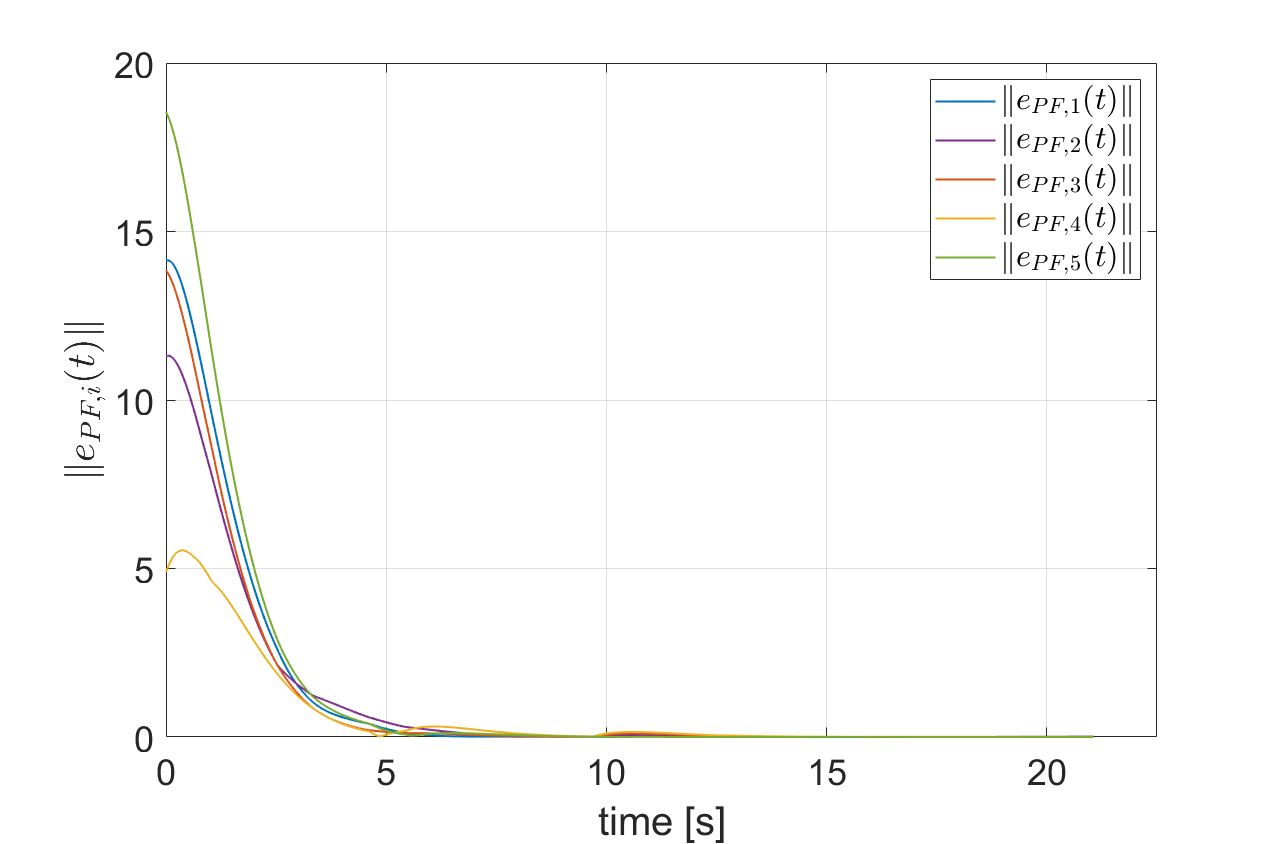}
    \caption{Path-following errors.}
    \label{fig:e_pf}
\end{figure}

\begin{figure} [h!]
    \centering
    \includegraphics[width = 1.00\linewidth]{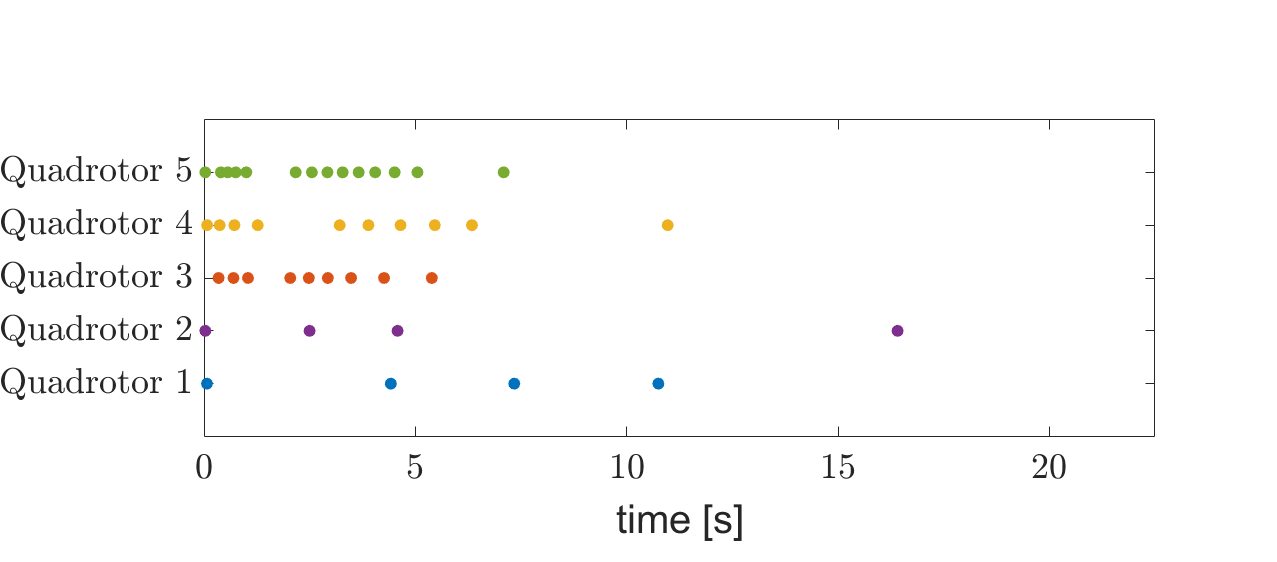}
    \caption{Event-triggered time instances.}
    \label{fig:event}
\end{figure}
\section{CONCLUSION} \label{VI}
This paper proposed a novel time-coordination algorithm using an event-triggered communication strategy. The exponential convergence of the coordination errors to a neighborhood of zero was proven using an ISS framework. Simulation results validated that coordinating the path following of quadrotors is possible with intermittent event-triggered communications.



\bibliographystyle{ieeetr}
\bibliography{references}

\begin{thebibliography}{10}

\bibitem{Lee2013}
T.~Lee, K.~Sreenath, and V.~Kumar, ``Geometric control of cooperating multiple
  quadrotor uavs with a suspended payload,'' in {\em 52nd IEEE Conference on
  Decision and Control}, pp.~5510--5515, 2013.

\bibitem{Lee2017}
H.~Lee and H.~J. Kim, ``Constraint-based cooperative control of multiple aerial
  manipulators for handling an unknown payload,'' {\em IEEE Transactions on
  Industrial Informatics}, vol.~13, no.~6, pp.~2780--2790, 2017.

\bibitem{Morgan20141725}
D.~Morgan, S.-J. Chung, and F.~Y. Hadaegh, ``Model predictive control of swarms
  of spacecraft using sequential convex programming,'' {\em Journal of
  Guidance, Control, and Dynamics}, vol.~37, no.~6, p.~1725–1740, 2014.

\bibitem{Bandyopadhyay2015}
S.~Bandyopadhyay, G.~P. Subramanian, R.~Foust, D.~Morgan, S.-J. Chung, and
  F.~Y. Hadaegh, ``A review of impending small satellite formation flying
  missions,'' {\em 53rd AIAA Aerospace Sciences Meeting}, 2015.

\bibitem{Schmuck2017}
P.~Schmuck and M.~Chli, ``Multi-uav collaborative monocular slam,'' in {\em
  2017 IEEE International Conference on Robotics and Automation (ICRA)},
  pp.~3863--3870, 2017.

\bibitem{Dubé2017}
R.~Dubé, A.~Gawel, H.~Sommer, J.~Nieto, R.~Siegwart, and C.~Cadena, ``An
  online multi-robot slam system for 3d lidars,'' in {\em 2017 IEEE/RSJ
  International Conference on Intelligent Robots and Systems (IROS)},
  pp.~1004--1011, 2017.

\bibitem{Choe20161744}
R.~Choe, J.~Puig-Navarro, V.~Cichella, E.~Xargay, and N.~Hovakimyan,
  ``Cooperative trajectory generation using pythagorean hodograph bézier
  curves,'' {\em Journal of Guidance, Control, and Dynamics}, vol.~39, no.~8,
  p.~1744 – 1763, 2016.

\bibitem{Cichella2021}
V.~Cichella, I.~Kaminer, C.~Walton, N.~Hovakimyan, and A.~M. Pascoal, ``Optimal
  multivehicle motion planning using bernstein approximants,'' {\em IEEE
  Transactions on Automatic Control}, vol.~66, no.~4, pp.~1453--1467, 2021.

\bibitem{CICHELLA201313}
V.~Cichella, R.~Choe, S.~B. Mehdi, E.~Xargay, N.~Hovakimyan, I.~Kaminer, and
  V.~Dobrokhodov, ``A 3d path-following approach for a multirotor uav on
  so(3),'' {\em IFAC Proceedings Volumes}, vol.~46, no.~30, pp.~13--18, 2013.
\newblock 2nd IFAC Workshop on Research, Education and Development of Unmanned
  Aerial Systems.

\bibitem{Kaminer2006}
I.~Kaminer, O.~Yakimenko, A.~Pascoal, and R.~Ghabcheloo, ``Path generation,
  path following and coordinated control for timecritical missions of multiple
  uavs,'' in {\em 2006 American Control Conference}, pp.~4906--4913, 2006.

\bibitem{Ghabcheloo2007133}
R.~Ghabcheloo, A.~Pascoal, C.~Silvestre, and I.~Kaminer, ``Non-linear
  co-ordinated path following control of multiple wheeled robots with
  bidirectional communication constraints,'' {\em International Journal of
  Adaptive Control and Signal Processing}, vol.~21, no.~2-3, p.~133–157,
  2007.

\bibitem{Xargay2013499}
E.~Xargay, I.~Kaminer, A.~Pascoal, N.~Hovakimyan, V.~Dobrokhodov, V.~Cichella,
  A.~Aguiar, and R.~Ghabcheloo, ``Time-critical cooperative path following of
  multiple unmanned aerial vehicles over time-varying networks,'' {\em Journal
  of Guidance, Control, and Dynamics}, vol.~36, no.~2, p.~499–516, 2013.

\bibitem{Cichella2015945}
V.~Cichella, I.~Kaminer, V.~Dobrokhodov, E.~Xargay, R.~Choe, N.~Hovakimyan,
  A.~P. Aguiar, and A.~M. Pascoal, ``Cooperative path following of multiple
  multirotors over time-varying networks,'' {\em IEEE Transactions on
  Automation Science and Engineering}, vol.~12, no.~3, p.~945–957, 2015.

\bibitem{Mehdi2017}
S.~B. Mehdi, V.~Cichella, T.~Marinho, and N.~Hovakimyan, ``Collision avoidance
  in multi-vehicle cooperative missions using speed adjustment,'' in {\em 2017
  IEEE 56th Annual Conference on Decision and Control (CDC)}, pp.~2152--2157,
  2017.

\bibitem{Tabasso2020436}
C.~Tabasso, V.~Cichella, S.~Bilal~Mehdi, T.~Marinho, and N.~Hovakimyan,
  ``Guaranteed collision avoidance in multivehicle cooperative missions using
  speed adjustment,'' {\em Journal of Aerospace Information Systems}, vol.~17,
  no.~8, p.~436–453, 2020.

\bibitem{Tabasso2022704}
C.~Tabasso, C.~Kielas-Jensen, V.~Cichella, S.~Manyam, D.~W. Casbeer, and
  I.~Weintraub, ``Continuous monitoring of a path-constrained moving target by
  multiple unmanned aerial vehicles,'' {\em Journal of Guidance, Control, and
  Dynamics}, vol.~45, no.~4, p.~704–713, 2022.

\bibitem{Hyungsoo2023}
H.~Kang, I.~Kaminer, V.~Cichella, and N.~Hovakimyan, ``Coordinated path
  following of quadrotors over time-varying digraphs connected in an integral
  sense,'' {\em arXiv preprint}, 2023.

\bibitem{Dimarogonas2012}
D.~V. Dimarogonas, E.~Frazzoli, and K.~H. Johansson, ``Distributed
  event-triggered control for multi-agent systems,'' {\em IEEE Transactions on
  Automatic Control}, vol.~57, no.~5, pp.~1291--1297, 2012.

\bibitem{SEYBOTH2013245}
G.~S. Seyboth, D.~V. Dimarogonas, and K.~H. Johansson, ``Event-based
  broadcasting for multi-agent average consensus,'' {\em Automatica}, vol.~49,
  no.~1, pp.~245--252, 2013.

\bibitem{Fan2015}
Y.~Fan, L.~Liu, G.~Feng, and Y.~Wang, ``Self-triggered consensus for
  multi-agent systems with zeno-free triggers,'' {\em IEEE Transactions on
  Automatic Control}, vol.~60, no.~10, pp.~2779--2784, 2015.

\bibitem{YANG2019129}
R.~Yang, H.~Zhang, G.~Feng, H.~Yan, and Z.~Wang, ``Robust cooperative output
  regulation of multi-agent systems via adaptive event-triggered control,''
  {\em Automatica}, vol.~102, pp.~129--136, 2019.

\bibitem{Hu2018}
W.~Hu, L.~Liu, and G.~Feng, ``Cooperative output regulation of linear
  multi-agent systems by intermittent communication: A unified framework of
  time- and event-triggering strategies,'' {\em IEEE Transactions on Automatic
  Control}, vol.~63, no.~2, pp.~548--555, 2018.

\bibitem{Qian2019}
Y.-Y. Qian, L.~Liu, and G.~Feng, ``Output consensus of heterogeneous linear
  multi-agent systems with adaptive event-triggered control,'' {\em IEEE
  Transactions on Automatic Control}, vol.~64, no.~6, pp.~2606--2613, 2019.

\bibitem{Wu2018}
Z.-G. Wu, Y.~Xu, R.~Lu, Y.~Wu, and T.~Huang, ``Event-triggered control for
  consensus of multiagent systems with fixed/switching topologies,'' {\em IEEE
  Transactions on Systems, Man, and Cybernetics: Systems}, vol.~48, no.~10,
  pp.~1736--1746, 2018.

\bibitem{Cheng2019}
B.~Cheng, X.~Wang, and Z.~Li, ``Event-triggered consensus of homogeneous and
  heterogeneous multiagent systems with jointly connected switching
  topologies,'' {\em IEEE Transactions on Cybernetics}, vol.~49, no.~12,
  pp.~4421--4430, 2019.

\bibitem{Hu2019}
W.~Hu, L.~Liu, and G.~Feng, ``Event-triggered cooperative output regulation of
  linear multi-agent systems under jointly connected topologies,'' {\em IEEE
  Transactions on Automatic Control}, vol.~64, no.~3, pp.~1317--1322, 2019.

\bibitem{Jia2018}
Q.~Jia and W.~K.~S. Tang, ``Consensus of multi-agents with event-based
  nonlinear coupling over time-varying digraphs,'' {\em IEEE Transactions on
  Circuits and Systems II: Express Briefs}, vol.~65, no.~12, pp.~1969--1973,
  2018.

\bibitem{HAN2015196}
Y.~Han, W.~Lu, and T.~Chen, ``Consensus analysis of networks with time-varying
  topology and event-triggered diffusions,'' {\em Neural Networks}, vol.~71,
  pp.~196--203, 2015.

\bibitem{Hao2023}
Y.~Hao, L.~Liu, and G.~Feng, ``Event-triggered cooperative output regulation of
  heterogeneous multiagent systems under switching directed topologies,'' {\em
  IEEE Transactions on Cybernetics}, vol.~53, no.~2, pp.~1026--1038, 2023.

\bibitem{phdenric2013}
E.~Xargay, {\em Time-Critical Cooperative Path-Following Control of Multiple
  Unmanned Aerial Vehicles}.
\newblock PhD thesis, University of Illinois at Urbana-Champaign, 2013.

\bibitem{kha2002}
H.~K. Khalil, {\em Nonlinear Systems}.
\newblock Prentice-Hall, Englewood Cliffs, NJ, 2002.

\end{thebibliography}

\end{document}